\documentclass[a4paper,12pt]{elsarticle}
\usepackage[T1]{fontenc}
\usepackage[english]{babel}
\usepackage{ae,aecompl}
\usepackage{amsmath, amsthm}
\usepackage{amssymb}
\usepackage[utf8]{inputenc}
\usepackage[section] {placeins}
\usepackage[ruled, vlined, linesnumbered]{algorithm2e}
\newtheorem {Theorem}                 {Theorem}         [section]
\newtheorem {theorem}      [Theorem]  {Theorem}
\newtheorem {myalgorithm}    [Theorem]  {Algorithm}

\newtheorem {lemma}        [Theorem]  {Lemma}

\input amssym.def
\input amssym

\usepackage{pgfplots}
\usepackage{verbatim}
\usepackage{subfigure}
\usepackage{tikz}
\usetikzlibrary{shapes,arrows,backgrounds,mindmap}
\usepackage{titlesec}
\journal{arXiv}

\address{}

\begin{document} 
\begin{frontmatter}
\title{$2$-blocks in strongly biconnected directed graphs}
\author{Raed Jaberi}
\begin{abstract}  

	A directed graph $G=(V,E)$ is called strongly biconnected if $G$ is strongly connected and the underlying graph of $G$ is biconnected. A strongly biconnected component of a strongly connected graph $G=(V,E)$ is a maximal vertex subset $L\subseteq V$ such that the induced subgraph on $L$ is strongly biconnected. Let $G=(V,E)$ be a strongly biconnected directed graph. A $2$-edge-biconnected block in $G$ is a maximal vertex subset $U\subseteq V$ such that for any two distict vertices $v,w \in U$ and for each edge $b\in E$, the vertices $v,w$ are in the same strongly biconnected components of $G\setminus\left\lbrace b\right\rbrace $. A $2$-strong-biconnected block in $G$ is a maximal vertex subset $U\subseteq V$ of size at least $2$ such that for every pair of distinct vertices $v,w\in U$ and for every vertex $z\in V\setminus\left\lbrace v,w \right\rbrace $, the vertices $v$ and $w$ are in the same strongly biconnected component of $G\setminus \left\lbrace v,w \right\rbrace $. In this paper we study $2$-edge-biconnected blocks and $2$-strong biconnected blocks.
	 
\end{abstract} 
\end{frontmatter}
\section{Introduction}
Let $G=(V,E)$ be a directed graph. A $2$-edge block of $G$ is a maximal vertex subset $L_e \subseteq V$ with $|L_e|>1$ such that for each pair of distinct vertices $x,y\in L_e$, $G$ contains two edge-disjoint paths from $x$ to $y$ and two edge-disjoint paths from $y$ to $x$. A $2$-strong block of $G$ is a maximal vertex subset $B_s\subseteq V$ with $|B_s|>1$ such that for each  pair of distinct vertices $x,y\in B_s$ and for every vertex $w\in V\setminus\left\lbrace x,y\right\rbrace $, $x$ and $y$ belong to the same strongly connected component of $G\setminus \left\lbrace w \right\rbrace $. $G$ is called strongly biconnected if $G$ is strongly connected and the underlying graph of $G$ is biconnected. This class of directed graphs was introduced by Wu and Grumbach \cite{WG2010}. A strongly biconnected component of $G$ is a maximal vertex subset $C\subseteq V$ such that the induced subgraph on $C$ is strongly biconnected \cite{WG2010}. Let $G=(V,E)$ be a strongly biconnected directed graph. An edge $e\in E$ is a b-bridge if the subgraph $G\setminus \left\lbrace e\right\rbrace =(V,E\setminus \left\lbrace  e\right\rbrace) $ is not strongly biconnected. A vertex $w\in V$ is a b-articulation point if  $G\setminus \left\lbrace w\right\rbrace$ is not strongly biconnected, where $G\setminus \left\lbrace w\right\rbrace$ is the subgraph obtained from $G$ by deleting $w$. $G$ is $2$-edge-strongly-biconnected (respectively, $2$-vertex-strongly biconnected) if $|V|>2$ and $G$ has no b-bridges (respectively, b-articulation points). A $2$-edge-biconnected block in $G$ is a maximal vertex subset $U\subseteq V$ such that for any two distict vertices $v,w \in U$ and for each edge $b\in E$, the vertices $v,w$ are in the same strongly biconnected components of $G\setminus\left\lbrace b\right\rbrace $. The $2$-edge blocks of $G$ are disjoint \cite{Jaberi15}. Notice that $2$-edge-biconnected blocks are not necessarily disjoint, as shown in Figure \ref{fig:blockssb}.
 \begin{figure}[htp]
	\centering
	
		\begin{tikzpicture}[xscale=2]
		\tikzstyle{every node}=[color=black,draw,circle,minimum size=0.9cm]
		\node (v1) at  (3.6,4){$1$};
		\node (v2) at  (4,0){$2$};
		\node (v3) at (1, 5) {$3$};
		\node (v4) at (2,5.6) {$4$};
		\node (v5) at (3,1) {$5$};
		\node (v6) at (3.5,2.5) {$6$};
		\node (v7) at  (0.5,7){$7$};
		\node (v8) at (5,4.5) {$8$};
		\node (v9) at  (-1,2.5){$9$};
		\node (v10)at  (2,1.4){$10$};
		\node (v11) at (1,1.4) {$11$};
		\node (v12) at  (2.5,-2){$12$};
        \node (v13) at  (1,-3){$13$};
        \node (v14) at (-1,0) {$14$};
        \node (v15) at (5,2.5) {$15$};
        \node (v16) at (3.7,7) {$16$};
		\begin{scope}   
		\tikzstyle{every node}=[auto=right]   
		\draw [-triangle 45] (v12) to [bend left ] (v13);
		\draw [-triangle 45] (v10) to  (v13);
		\draw [-triangle 45] (v13) to (v9);
			\draw [-triangle 45] (v11) to (v12);
		\draw [-triangle 45] (v14) to (v11);
		\draw [-triangle 45] (v9) to   (v11);
		\draw [-triangle 45] (v11) to  (v10);
		\draw [-triangle 45] (v13) to [bend left ]  (v14);
		\draw [-triangle 45] (v14) to  [bend left ] (v9);
		\draw [-triangle 45] (v10) to  (v12);
		
			\draw [-triangle 45] (v12) to  (v14);
			\draw [-triangle 45] (v9) to  [bend left ] (v10);
			\draw [-triangle 45] (v9) to  (v7);
			\draw [-triangle 45] (v7) to  (v4);
			\draw [-triangle 45] (v10) to  (v5);
			\draw [-triangle 45] (v5) to  (v4);
			\draw [-triangle 45] (v4) to  (v2);
			\draw [-triangle 45] (v4) to  (v6);
			\draw [-triangle 45] (v4) to  (v3);
			\draw [-triangle 45] (v3) to  (v14);
			\draw [-triangle 45] (v6) to  (v15);
			
			\draw [-triangle 45] (v2) to  (v15);
			\draw [-triangle 45] (v4) to  (v1);
			\draw [-triangle 45] (v1) to  (v15);
			\draw [-triangle 45] (v15) to  (v8);
			\draw [-triangle 45] (v8) to  (v4);
			\draw [-triangle 45] (v15) to  (v16);
			\draw [-triangle 45] (v16) to  (v4);
			\draw [-triangle 45] (v2) to  (v12);
		\end{scope}
		\end{tikzpicture}
	\caption{A strongly biconnected directed graph $G$. The vertex subset $\left\lbrace 9,10,11,12,13,14,4,15 \right\rbrace $ is a $2$-edge block of $G$. Notice that the vertices $15,12$ are not in the same $2$-edge-biconnected block because $15$ and $12$ are not in the same strongly biconnected component of $G\setminus\left\lbrace (2,15) \right\rbrace $. Moreover, $G$ has two $2$-edge-biconnected blocks $U_1=\left\lbrace 9,10,11,12,13,14,4 \right\rbrace $ and $U_2=\left\lbrace 4,15 \right\rbrace $.  $U_1$ and $U_2$ share vertex $4$}  
	\label{fig:blockssb} 
\end{figure}
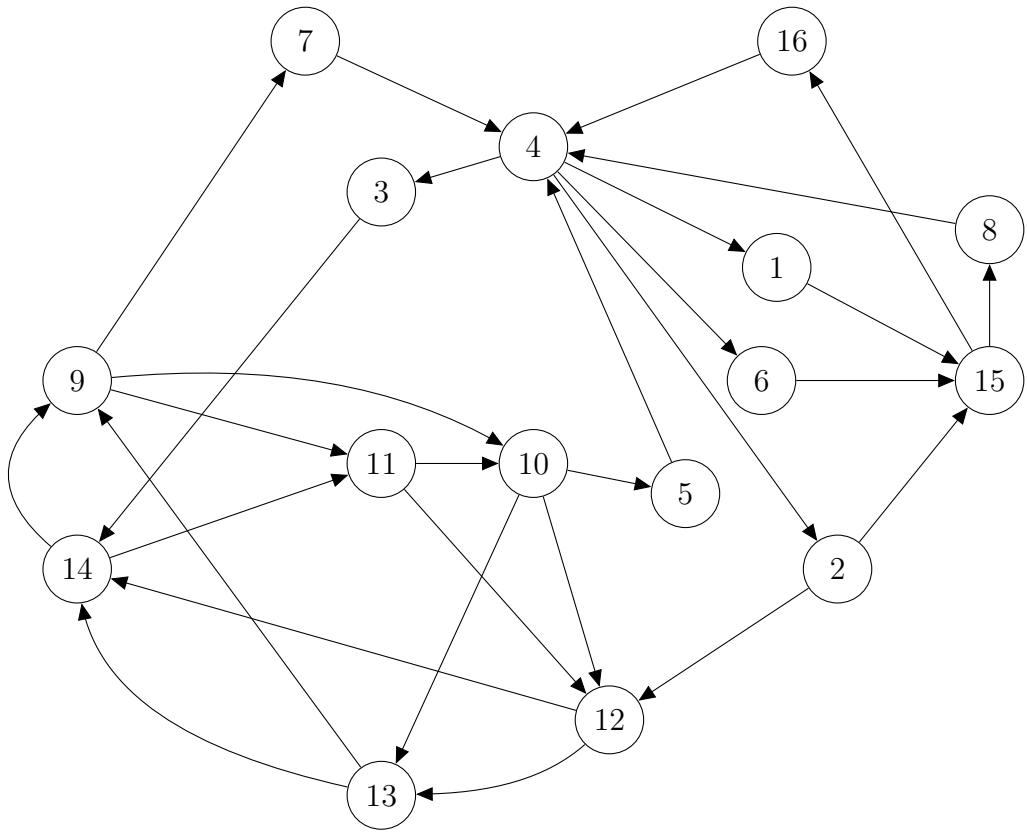

 A $2$-strong-biconnected block in $G$ is a maximal vertex subset $U\subseteq V$ of size at least $2$ such that for every pair of distinct vertices $v,w\in U$ and for every vertex $z\in V\setminus\left\lbrace v,w \right\rbrace $, the vertices $v$ and $w$ are in the same strongly biconnected component of $G\setminus \left\lbrace z \right\rbrace $.

Blocks, articulation points, and bridges of an undirected graph can be calculated in $O(n+m)$ time \cite{T72,T74,JS13}. In\cite{G10},  Georgiadis presented a linear time algorithm to test whether a directed graph is $2$-vertex-connected.  Strong articulation points and strong bridges of a directed graph can be computed in $O(n+m)$ time \cite{ILS12,FILOS12}. Jaberi \cite{Jaberi15} presented algorithms for computing $2$-strong blocks, and $2$-edge blocks of a directed graph.  Georgiadis et al. \cite{GILP14SODA,GILP14VertexConnectivity} gave linear time algorithms for determining $2$-edge blocks and $2$-strong blocks. 
Wu and Grumbach \cite{WG2010} introduced the concept of strongly biconnected directed graphs and the concept of strongly biconnected components. Jaberi \cite{Jaberi2020} studied b-bridges in strongly biconnected directed graphs. In this paper we study $2$-edge-biconnected blocks and $2$-strong biconnected blocks.

\section{$2$-edge-biconnected blocks}
In this section we study $2$-edge-biconnected blocks and present an algorithm for computing them. 
Let $G=(V,E)$ be a strongly biconnected directed graph. For every pair of distinct vertices $x,y\in V$, we write $x \overset{e}{\leftrightsquigarrow } y$ if for any edge $b\in E$, the vertices $x,y$ belong to the same strongly biconnected component of $G\setminus\lbrace b\rbrace$. A $2$-edge-biconnected blocks in $G$ is a maximal subset of vertices $ U\subseteq V$ with $|U|>1$ such that for any two vertices $x,y \in U$, we have $x \overset{e}{\leftrightsquigarrow } y$. A $2$-edge-strongly-biconnected component in $G$ is a maximal vertex subset $C_{2eb}\subseteq V$ such that the induced subgraph on $C_{2eb}$ is $2$-edge-strongly biconnected. Note that the strongly biconnected directed graph in Figure \ref{fig:blockssb} contains one $2$-edge-strongly biconnected component $\left\lbrace 9,10,11,12,13,14 \right\rbrace $, which is a subset of the $2$-edge-biconnected block $\left\lbrace 9,10,11,12,13,14,4 \right\rbrace $.

\begin{lemma}
Let $G=(V,E)$ be a strongly biconnected directed graph and let $C_{2eb}$ be a $2$-edge-strongly biconnected component of $G$. Then  $C_{2eb}$ is a subset of a $2$-edge-biconnected block of $G$. 
\end{lemma}
\begin{proof}
	Let $x$ and $y$ be distinct vertices in $C_{2eb}$ and let $e\in E$. Let $G[C_{2eb}]$ be the induced subgraph on $C_{2eb}$. 
	By definition, the subgraph obtained from $G[C_{2eb}]$ by deleting $e$ is still strongly biconnected. Therefore, we have $x \overset{e}{\leftrightsquigarrow } y$.
\end{proof}

$2$-edge blocks are disjoint \cite{Jaberi15}. Note that $2$-edge biconnected blocks are not necessarily disjoint. But any two of them share at most one vertex, as illustrated in Figure \ref{fig:blockssb}. 
\begin{lemma}\label{def:lemmaforsb}
	Let $U_1,U_2$ be distinct $2$-edge-biconnected blocks of a strongly biconnected directed graph $G=(V,E)$. Then $|U_1\cap U_2|\leq 1$
\end{lemma}
\begin{proof}
		Assume for the purpose of contradiction that  $|U_1\cap U_2|>1$.  Let $x \in U_1 \setminus (U_1\cap U_2)$ and let $y\in U_2 \setminus (U_1\cap U_2)$. Let $v,w \in U_1\cap U_2$ with $v\neq w$ and let $b\in E$. Notice that $x,v$ belong to the same stongly connected component of $G\setminus\left\lbrace b\right\rbrace $ since $x \overset{e}{\leftrightsquigarrow }v$. Moreover,  $v,y$ belong to the same stongly connected component of $G\setminus\left\lbrace b\right\rbrace $. Consequently, $x,y$ are in the same stongly connected component of $G\setminus\left\lbrace b\right\rbrace $. Then, the vertices $x,y$ do not lie in the same strongly biconnected component of $G\setminus\left\lbrace b\right\rbrace $.  Suppose that $C_x,C_y$ are two strongly biconnected components of $G\setminus\left\lbrace b\right\rbrace $ such that $x\in C_x$ and $y \in C_y$. There are two cases to consider.
		\begin{enumerate}
			\item $v \in C_x \cap C_y $. In this case $w\notin C_x \cap C_y $. Suppose without loss of generality that $w \in C_x$. Then $w,y$ are not in the same stongly biconnected components of $G\setminus\left\lbrace b\right\rbrace $. But this contradicts that $w \overset{e}{\leftrightsquigarrow }y$
			\item $v \notin C_x \cap C_y $. Suppose without loss of generality that $v \in C_x$. Then $v,y$ do not lie in the same strongly biconnected component of $G\setminus\left\lbrace b\right\rbrace $. But this contradicts that $v \overset{e}{\leftrightsquigarrow }y$
		\end{enumerate} 	
\end{proof}
Using similar arguments as in Lemma \ref{def:lemmaforsb}, we can prove the following.
\begin{lemma} \label{def:lemmafor2strongbiconnectedblocks}
		Let $G=(V,E)$ be a strongly biconnected directed graph and let $\left\lbrace w_0,w_1,\ldots,w_t\right\rbrace \subseteq V $ such that $w_0 \overset{e}{\leftrightsquigarrow } w_t$ and $w_{i-1} \overset{e}{\leftrightsquigarrow } w_i$ for $i \in\lbrace 1,2\ldots,t\rbrace$. Then $\left\lbrace w_0,w_1,\ldots,w_t\right\rbrace$ is a subset of a $2$-edge biconnected block of $G$.
\end{lemma}

\begin{lemma}\label{def:verticesarenotbbridges}
Let $G=(V,E)$ be a strongly biconnected directed graph and let $v,w$ be two distinct vertices in $G$. Let $b$ be an edge in $G$ such that $b$ is not a b-bridge. Then, the vertices $v,w$ are in the same strongly biconnected component of $G\setminus\left\lbrace
 b \right\rbrace $ 
\end{lemma}
\begin{proof}
	immediate from definition.
\end{proof}
 Algorithm \ref{algo:algorrithmfor2edgebiconnectedblocks} shows an algorithm for computing all the $2$-edge biconnected blocks of a strongly biconnected directed graph.
\begin{figure}[h]
	\begin{myalgorithm}\label{algo:algorrithmfor2edgebiconnectedblocks}\rm\quad\\[-5ex]
		\begin{tabbing}
			\quad\quad\=\quad\=\quad\=\quad\=\quad\=\quad\=\quad\=\quad\=\quad\=\kill
			\textbf{Input:} A strongly biconnected directed graph $G=(V,E)$.\\
			\textbf{Output:} The $2$-edge biconnected blocks of $G$.\\
			{\small 1}\> Compute the b-bridges of $G$\\
			{\small 2}\> \textbf{If} $G$ has no b-bridges \textbf{then}.\\
			{\small 3}\>\> Output $V$.\\
			{\small 4}\> \textbf{else}\\
			{\small 5}\>\> Let $L$ be an $n\times n$ matrix.\\
			{\small 6}\>\> Initialize $L$ with $1$s.\\
			{\small 7}\>\> \textbf{for} each  b--bridge $b$ of $G$ \textbf{do} \\
			{\small 8}\>\>\> calculate the strongly biconnected components of  $G\setminus \lbrace b\rbrace$ \\
			{\small 9}\>\>\> \textbf{for} each pair $(x,y) \in V\times V$ \textbf{do} \\
			{\small 10}\>\>\>\> \textbf{if} $x,y$ in different strongly biconnected components of $G\setminus \lbrace b\rbrace$ \textbf{then}\\
			{\small 11}\>\>\>\>\> $L[x,y] \leftarrow 0$. \\
			{\small 12}\>\> $E^{eb} \leftarrow \emptyset$. \\
			{\small 13}\>\> \textbf{for} every pair $(x,y) \in V\times V $ \textbf{do} \\
			{\small 14}\>\>\> \textbf{if} $L[x,y]=1$ and $L[y,x]=1$ \textbf{then} \\
			{\small 15}\>\>\>\> $E^{eb} \leftarrow E^{eb}\cup\left\lbrace(x,y) \right\rbrace$  \\
			{\small 16}\>\> $G^{eb}\leftarrow (V,E^{eb})$\\
			{\small 17}\>\>  Compute all the blocks of size $\geq 2$ in  $G^{eb}$ and output them. 
		\end{tabbing}
	\end{myalgorithm}
\end{figure}

The correctness of this algorithm follows from Lemma \ref{def:lemmaforsb}, Lemma \ref{def:lemmafor2strongbiconnectedblocks}, and Lemma \ref{def:verticesarenotbbridges}.

\begin{theorem}
	Algorithm \ref{algo:algorrithmfor2edgebiconnectedblocks} runs in  $O(n^{3})$ time.
\end{theorem}
\begin{proof}
The b-bridges of $G$ can be computed in $O(nm)$ time \cite{Jaberi2020}. Strongly biconnected components can be calculated in linear time \cite{WG2010}. Lines $7$--$11$ take $O(b.n^{2})$, where $b$ is the number of b-bridges in $G$. The time required for building $G^{eb}$ in lines $12$--$16$ is $O(n^{2})$. Moreover, the blocks of an undirected graph can be found in linear time using Tarjan's algorithm. \cite{T72,JS13}..
\end{proof}
\section{$2$-strong-biconnected blocks} 
In this section we illustrate some properties of $2$-strong-biconnected blocks. The strongly biconnected directed graph in Figure \ref{fig:stronglybiconnectedblocksdirectedgraph} has two $2$-strong biconnected blocks $ L_1=\left\lbrace 1,2,3,4 \right\rbrace $ and $L_2=\left\lbrace 3,4.5.6 \right\rbrace $. Note that $L_1$ and $L_2$ share two vertices. The intersection of any two distinct $2$-strong biconnected blocks contains at most $2$ vertices. Note also that the subgraph induced by the $2$-strong biconnected block $ L_1$ has no edges.
\begin{figure}[htp]
	\centering
	
	\begin{tikzpicture}[xscale=2]
	\tikzstyle{every node}=[color=black,draw,circle,minimum size=0.9cm]
	\node (v1) at  (0.5,0.4){$1$};
	\node (v2) at  (4,0){$2$};
	\node (v3) at (0.5, 6) {$3$};
	\node (v4) at (4,5.4) {$4$};
	\node (v5) at (0.5,12) {$5$};
	\node (v6) at (4,10.5) {$6$};
	\node (v7) at  (2,1.2){$7$};
	\node (v8) at (2,0) {$8$};
	\node (v9) at  (2,7){$9$};
	\node (v10)at  (2,5){$10$};
	\node (v11) at (2,13) {$11$};
	\node (v12) at  (2,10.5){$12$};
	\node (v13) at  (0,3.5){$13$};
	\node (v14) at (1,3.5) {$14$};
	\node (v15) at (4,2.5) {$15$};
	\node (v16) at (5,2) {$16$};
	\node (v17) at  (0,9){$17$};
	\node (v18) at (1,9) {$18$};
	\node (v19) at (4,7.5) {$19$};
	\node (v20) at (5,7) {$20$};
	
	\begin{scope}   
	\tikzstyle{every node}=[auto=right]   
	\draw [-triangle 45] (v1) to (v7);
	\draw [-triangle 45] (v7) to (v2);
	\draw [-triangle 45] (v2) to (v8);
	\draw [-triangle 45] (v8) to (v1);
	\draw [-triangle 45] (v3) to (v9);
	\draw [-triangle 45] (v9) to (v4);
	\draw [-triangle 45] (v4) to (v10);
	\draw [-triangle 45] (v10) to (v3);
	\draw [-triangle 45] (v5) to (v11);
	\draw [-triangle 45] (v11) to (v6);
	\draw [-triangle 45] (v6) to (v12);
	\draw [-triangle 45] (v12) to (v5);
	
	\draw [-triangle 45] (v1) to (v13);
	\draw [-triangle 45] (v13) to (v3);
	\draw [-triangle 45] (v3) to (v14);
	\draw [-triangle 45] (v14) to (v1);
	\draw [-triangle 45] (v4) to  (v1);
	\draw [-triangle 45] (v3) to (v2);
	\draw [-triangle 45] (v2) to (v16);
	\draw [-triangle 45] (v16) to (v4);
	\draw [-triangle 45] (v4) to (v15);
	\draw [-triangle 45] (v15) to (v2);
	
	\draw [-triangle 45] (v3) to (v17);
	\draw [-triangle 45] (v17) to (v5);
	\draw [-triangle 45] (v5) to (v18);
	\draw [-triangle 45] (v18) to (v3);
	\draw [-triangle 45] (v3) to  (v6);
	\draw [-triangle 45] (v4) to (v5);
	\draw [-triangle 45] (v4) to (v19);
	\draw [-triangle 45] (v19) to (v6);
	\draw [-triangle 45] (v6) to (v20);
	\draw [-triangle 45] (v20) to (v4);

	\end{scope}
	\end{tikzpicture}
	\caption{A strongly biconnected directed graph $G=(V,E)$. The vertices $2$ and $6$ are in the same $2$-strong block of $G$ but they do not belong to the same $2$-strong-biconnected-block of $G$ since they are not in the same strongly biconnected component of $G\setminus\left\lbrace  4\right\rbrace $. Moreover, $G$ has two $2$-strong biconnected blocks $ L_1=\left\lbrace 1,2,3,4 \right\rbrace $ and $L_2=\left\lbrace 3,4.5.6 \right\rbrace $. Note that  $|L_1\cap L_2|=2$  }
	\label{fig:stronglybiconnectedblocksdirectedgraph}
\end{figure}
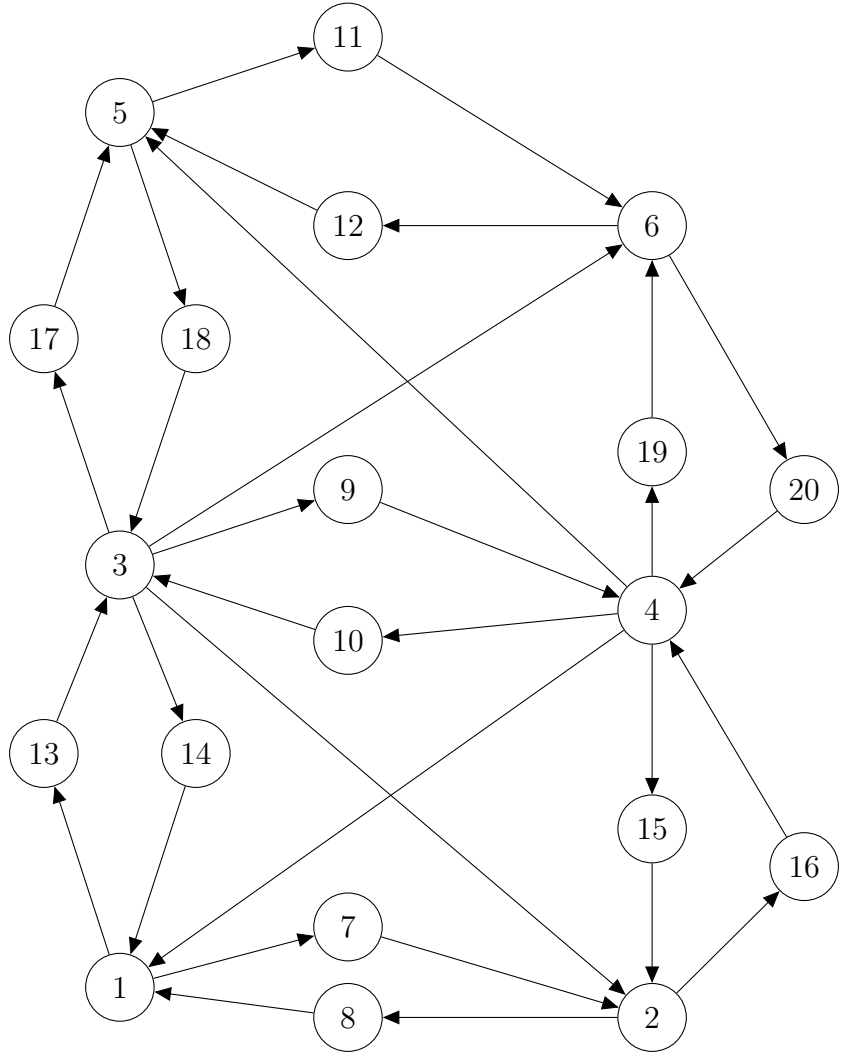

Let $G=(V,E)$ be a strongly biconnected directed graph. A $2$-vertex-strongly biconnected component $C_{2sb}$ is a maximal vertex subset $C_{2sb}\subseteq V$ such that the induced subgraph on $C_{2sb}$ is $2$-vertex-strongly biconnected. Each $2$-vertex-strongly biconnected component $C_{2sb}$ of $G$ is a subset of a $2$-strong-biconnected-block of $G$. Furthermore, each $2$-vertex-strongly biconnected component $C_{2sb}$ of $G$ is $2$-vertex connected. Therefore, the subgraph induced by $C_{2sb}$ contains at least $2|C_{2sb}|$ edges. In contrast to, the subgraphs induced by the $2$-strong-biconnected blocks do not necessarily contain edges.

\end{document}